\documentclass{llncs}

\usepackage{cite}
\usepackage{amssymb}
\usepackage{amsfonts}
\usepackage{amsmath}
\usepackage{tikz}
\usetikzlibrary{calc}
\usepackage{enumitem}
\usepackage{xspace}
\usepackage{todonotes}

\newcommand{\probDef}[3]{
  \begin{quote}
   \textsc{#1}\\
  \textbf{Input:} #2\\
  \textbf{Question:} #3
  \end{quote}
}

\newcommand{\score}{\mathrm{score}}

\newcommand{\para}[1]{\medskip\noindent\textbf{#1}\quad}
\newcommand{\paranoskip}[1]{\noindent\textbf{#1}\quad}

\usepackage{thmtools}
\usepackage{thm-restate}

\begin{document}

\pagestyle{plain}

\title{Balanced Facilities on Random Graphs}
\titlerunning{Balanced Facilities on Random Graphs}

\author{Roee David, Nimrod Talmon}
\authorrunning{R. David and N. Talmon}
\institute{Weizmann Institute of Science \\
\email{roee.david@weizmann.ac.il, nimrodtalmon77@gmail.com}
}

\maketitle

\begin{abstract}
Given a graph $G$ with $n$ vertices and $k$ players,
each of which is placing a facility on one of the vertices of $G$,
we define the score of the $i$th player to be the number of vertices for which,
among all players, the facility placed by the $i$th player is the closest.
A placement is \emph{balanced} if all players get roughly the same score.
A graph is \emph{balanced} if all placements on it are balanced.
Viewing balancedness as a desired property in various scenarios,
in this paper we study balancedness properties of graphs,
concentrating on random graphs and on expanders.
We show that,
while both random graphs and expanders tend to have good balancedness properties,
random graphs are, in general, more balanced.
In addition,
we formulate and prove intractability of the combinatorial problem of deciding whether a given graph is balanced;
then,
building upon our analysis on random graphs and expanders,
we devise two efficient algorithms which,
with high probability,
generate balancedness certificates.
Our first algorithm is based on graph traversal,
while the other relies on spectral properties.
\end{abstract}

\section{Introduction}\label{section:introduction}

Consider a game played by $k$ players on some graph $G$.
The players place facilities on vertices of $G$ such that each player places one facility.
For each player we define a \emph{score},
defined as the number of vertices which,
among all other facilities, are closest to his or her facility;
ties are broken evenly,
such that if there are $z$ facilities closest to a vertex, then this vertex incurs a score increase of $1 / z$ to each of these facilities.
Such games are subject to extensive research;
some prominent study areas are
\emph{Voronoi games on graphs}~\cite{bandyapadhyay2015voronoi,durr2007nash,mavronicolas2008voronoi,teramoto2006voronoi}
and \emph{Competitive facility location games}~\cite{Friesz2007,saban2012competitive,banik2016discrete,banik2013optimal,ahn2004competitive}, where the players try to maximize their score.

In this paper,
however,
we concentrate on balancedness properties of such games,
thus consider having the score of the players be as close to each other as possible to be desired.
Indeed,
in some sense,
in this paper we take the point of view of the network designer
by studying balancedness properties of certain graphs.
To this end,
we say that a placement of $k$ facilities on a graph is \emph{$z$-balanced}
if all facilities get roughly the same score;
specifically,
a placement of $k$ facilities is \emph{$z$-balanced} if the score of each facility is at least $\lfloor n / k \rfloor - z$ and at most $\lceil n / k \rceil + z$.
We say further that a graph is \emph{$z$-balanced} if all placements on it are $z$-balanced.
A more formal definition is given in Section~\ref{section:preliminaries}.

\pagebreak

Graph balancedness,
besides being a natural and an interesting graph property from a combinatorial point of view,
is motivated by certain scenarios, two of which we briefly mention next.
As a first example,
consider a computer network to be built;
the network acts as the graph upon facilities, such as computer servers might be built.
It is of interest to have a network with good balancedness properties,
so that it will remain fair and efficient where such servers would be employed on top of it.
As a second example,
consider a design of a city to be built;
the city topology and, for instance, its roads,
act as the graph upon facilities, such as hospitals and child-care centers might be built.
It is of interest to have a city with good balancedness properties,
so it would be able to accommodate the needs of its future residents.
Indeed,
a city with bad balancedness properties might eventually become unpleasant and socially inferior.
Thus,
we believe that it is worthwhile to study balancedness properties of graphs,
as well as algorithms for verifying whether a given graph is balanced.

Some research has been done on balancedness of facilities in graphs,
including work on designing practical algorithms for finding balanced allocations~\cite{marin2011discrete}
and work considering balancedness in a perculation-like model~\cite{baroni2015fixed,van2015fixed}.
Other, different notions of balancedness in facility location games have been studied as well~\cite{honiden2009balancing}.
For an elaborate discussion on balancedness notions in facility location games, see~\cite{marsh1994equity}.

In this paper we analyze balancedness properties of certain graphs,
seeking to identify graphs which are balanced.
Specifically,
we concentrate on random graphs and also on expanders,
showing that these graphs usually have good balancedness properties.
Then,
building upon our analysis on random graphs and expanders,
we provide efficient algorithms for verifying whether a given graph is balanced.

\para{Initial Observations.}
As one of our goals is to identify graphs which have good balancedness properties,
let us identify certain such graphs.
As first examples,
observe that complete graphs and empty graphs are $0$-balanced for any number of players $k$;
indeed,
the score of each player is exactly $n / k$, for any placement of $k$ facilities on such graphs.
Both complete graphs and empty graphs are vertex-transitive graphs~\cite{godsil2013algebraic},
and we mention that any vertex-transitive graph is $0$-balanced for two players;
further,
a natural generalization of vertex-transitivity to sets of $k$ players yields graphs which are $0$-balanced also for $k$ players.

Naturally,
however,
not all graphs have good balancedness properties.
For example,
consider two players playing on the path graph $P_n$,
which is the graph with vertices $V = \{v_1, \ldots, v_n\}$ and edges $E = \{ \{ v_i, v_{i + 1} \} : i \in [n - 1]\}$.
Some placements of two facilities on the path graph $P_n$ are balanced,
for example where one facility is placed on $v_{n / 2}$ and the other facility is placed on $v_{n / 2 + 1}$:
  for $n \in \mathbb{N}_{\text{even}}$, such a placement is $0$-balanced,
  while for $n \in \mathbb{N}_{\text{odd}}$, such a placement is only $1$-balanced.
Some placements of two facilities on the path graph $P_n$,
however,
are not balanced,
for example where one facility is placed on $v_1$ and the other facility is placed on $v_2$:
  one player would have a score of $1$ while the other would have a score of $n - 1$.
This means that path graphs are not ($n - 2$)-balanced for two players,
which is, considering balancedness, ``the worst it can get''.

\para{Overview of the Paper.}
Preliminaries are provided in Section~\ref{section:preliminaries}.
Then,
motivated by our desire to identify graphs which are balanced and to further understand which factors influence graph balancedness,
in Section~\ref{section:randomgraphs} we consider random graphs and study their balancedness properties.
We show, in Theorem~\ref{thm:random graphs are balanced}, that random graphs have good balancedness properties.
Inspecting our proof of Theorem~\ref{thm:random graphs are balanced},
it looks as if what causes random graphs to be balanced is the fact that they are well-connected in a uniform way;
well-behaved (in the above-mentioned manner) graphs are usually referred to as expander graphs,
and thus,
in Section~\ref{section:expanders} we consider balancedness properties of expander graphs.
Specifically, we consider spectral expander graphs (for a precise definition, see Section~\ref{section:preliminaries}),
and mention that, with high probability, a random graph is a spectral expander
(see~\cite{feige2005spectral,furedi1981eigenvalues,friedman1989second}, e.g., for a proof of this fact).
In Theorem~\ref{thm:expanders are balanced},
we show that expanders have good balancedness properties.

Somewhat surprisingly,
though,
it turns out that,
with respect to their balancedness,
expanders are inferior to random graphs;
in particular,
for some values of the average degree of these graphs,
some random graphs are balanced while some expanders are not:
  in Theorem~\ref{thm:expander which is not balanced} we show an example for such expander which is not balanced.
This means that,
even though the expansion of random graphs influences their balancedness,
it is not sufficient,
and the inherent randomness of these graphs is also important for their balancedness.

In Section~\ref{section:algorithms} we consider the algorithmic problem of deciding whether a given graph is balanced.
We begin that section by proving that the corresponding combinatorial problem is intractable.
Then, building upon the analysis described in Sections~\ref{section:randomgraphs} (for random graphs) and~\ref{section:expanders} (for expander graphs),
we describe,
in Sections~\ref{section:algorithmone} and~\ref{section:algorithmtwo},
two efficient algorithms which,
given a graph,
provide a \emph{balancedness certificate}:
  in Section~\ref{section:algorithmone}
  we discuss an algorithm,
  based on graph traversal,
  which produces, in $O(n^2)$ time, a certificate that a graph is balanced;
  for random graphs, it produces such a certificate with high probability.
In Section~\ref{section:algorithmtwo}
we discuss a different algorithm,
based on spectral analysis,
which produces, in $O(d \cdot n \log n)$ (where $d$ is the average degree), a randomized certificate that a graph is balanced;
for random graphs, it produces such a certificate with high probability.
We conclude the paper in Section~\ref{section:outlook} with a discussion on directions for future research.

\section{Preliminaries}\label{section:preliminaries}

\paranoskip{General preliminaries.}
For $n \in \mathbb{N}$, we denote the set $\{1, \ldots, n\}$ by $[n]$.
Given a vector $\vec{v}\in\mathbb{R}^{n}$ we denote the $i$th coordinate of $\vec{v}$ by $(\vec{v})_{i}$.

\para{Graph theory and neighborhoods.}
Given a graph $G$ we denote by $V$ the vertex set of the graph and by $E$ the edge set of the graph.
For a set of vertices $S \subseteq V$,
we denote its complement by $\bar{S} := \{v \in V(G) : v \notin S\}$.
We denote the degree of a vertex $v \in V$ by $d_{v}$.
A \emph{$d$-regular} graph is a graph where $d_v = d$ for each $v \in V$,
while a \emph{roughly $d$-regular} graph is a graph where $d - o(d) \leq d_v \leq d + o(d)$ for each $v \in V$.
For a vertex $v \in V$ we denote by $N(v)$ the neighborhood of $v$, including $v$
(that is, $N(v) := \{u : \{u, v\} \in E(G)\} \cup \{v\}$).

For a set $S \subseteq V$,
$N(S)$ denotes the neighborhood of $S$;
that is, $N(S)=\cup_{s \in S}N(s)$
(indeed, with this definition, $S \subseteq N(S)$).
The $i$th \emph{neighborhood} of a vertex $v \in V$ is $N_{i}(v)=N(N_{i-1}(v))$,
where $N_{1}(v)=N(v)$.
The $i$th \emph{gained neighborhood} of a vertex $v \in V$
is $\tilde{N}_{i}(v)=N_{i}(v)\setminus N_{i-1}(v)$.
For a set of vertices $H \subseteq V$, let $G_{H}$ denote the subgraph of $G$ induced on the vertices of $H$.
For $S,T\subseteq V(G)$, $E_{G}(S,T)$ denotes
the number of edges between $S$ and $T$ in $G$,
where,
if $S,T$ are not
disjoint, then the edges in the induced subgraph of $S\cap T$ are
counted twice.

\para{Random graphs and expanders.}
A graph $G$ with $n$ vertices is a \emph{random graph} which is distributed by \emph{$G_{n,d}$},
and according to Erd\"{o}s-R\'{e}nyi model,
if each edge is included in the graph with probability $p=\frac{d}{n-1}$,
independently from every other edge.
We denote the adjacency matrix of a graph $G$ by $A_{G}$.
We denote $A_{G}$'s normalized eigenvectors by $e_{1}(G),e_{2}(G),\ldots,e_{n}(G)$
and the corresponding real eigenvalues by $\lambda_{1}(G)\ge\lambda_{2}(G)\ge,\ldots,\ge\lambda_{n}(G)$.
A graph $G$ is an \emph{$(s,\alpha)$ vertex expander} if for every $S\subseteq V(G)$ of size at most $s$
it holds that the neighborhood of $S$ is of size at least $\alpha|S|$.
A graph $G$ is a \emph{$\lambda$-expander} if $\max(\lambda_{2}(G),|\lambda_{n}(G)|) \le \lambda$.

\para{Facilities, scores, and balancedness.}
Given a graph $G$,
we consider $k$ players,
denoted by $P_1, \ldots, P_k$.
Each player places one facility on a vertex of $G$,
specifically player $P_i$ is choosing a vertex $u_i$ for placing his or her facility $f_i$.
We assume that no two players put their facilities on the same vertex.
Given such placement,
we define for a vertex $v \in V$ the set $C_v \subseteq \{P_1, \ldots, P_k\}$ of players whose facilities are placed closest to it;
that is,
$C_v := \{ P_i : d(v, u_i) \leq d(v, u_j), \forall j \neq i \}$,
where $d(u, v)$ stands for the distance between the vertices $u$ and $v$,
which is defined to be $\infty$ whenever $u$ and $v$ are disconnected.
We define the score of player $P_i$ to be
$\score(P_i) := \sum_{v \in V; P_i \in C_v} |C_v|^{-1}$.
A placement of $k$ facilities on $G$ is said to be $z$-\emph{balanced}
if $\lfloor n/k \rfloor - z \leq \score(P_i) \leq \lceil n/k \rceil + z$ for all $i \in [k]$.
A graph $G$ is said to be $z$-\emph{balanced} if any placement on it is $z$-balanced.

\section{Random Graphs}\label{section:randomgraphs}

In this section we consider random graphs,
generated according to Erd\"{o}s-R\'{e}nyi model (see Section~\ref{section:preliminaries}),
and analyse their balancedness.
In Theorem~\ref{thm:random graphs are balanced},
we prove some good balancedness properties of such graphs.
Building upon this theorem,
in Section~\ref{section:algorithmone} we describe an algorithm for verifying graph balancedness.

\begin{theorem}\label{thm:random graphs are balanced}		
  Let $G \sim G_{n,d}$ be a random graph,
  let $k = o(d)$,
  and let $d = \Omega (\log n)$.
  Then,
  for any arbitrarily small constant $\delta > 0$,
  with high probability,
  $G$ is $\delta n$-balanced for $k$ players.
\end{theorem}

\begin{proof}
For ease of presentation,
we present the proof for $k = 2$ players;
within the proof,
we explain how to generalize the proof to hold for $k = o(d)$ players,
as well as to why it does not hold for larger values of $k$.
Moreover,
we assume that $d = n^\epsilon$,
for an arbitrarily small constant $\epsilon$;
while the proof holds also for $d = \Omega(\log n)$,
it is easier to explain and clearer to understand for $d = n^\epsilon$,
and all ingredients of the proof are present in this case.

The proof builds upon the observation that the number of vertices with distance $i$ to any vertex in a $G_{n,d}$ is,
with high probability,
roughly $d^i$.
Then,
there are two cases to consider.
If $d \neq \Theta(n^{\frac{1}{i}})$ for some $i$ (Case 1),
then we show that,
with high probability,
the number of vertices at distance at most $i - 1$ from any facility is negligible,
while all vertices are at distance at most $i$ from any facility.
This means that most vertices are at the same distance from all facilities,
thus the graph is balanced in this case.
Otherwise,
if $d = \Theta(n^{\frac{1}{i}})$ for some $i$ (Case 2),
then we show that,
with high probability,
the number of vertices at distance at most $i - 1$ from any facility is $cn$,
where $0 \leq c \leq 1$ is some constant which is,
importantly,
equal for all facilities,
while all vertices are at distance at most $i$ from any facility.
This means that each facility has roughly $cn$ vertices which are the closest to it,
while most remaining vertices are shared between the facilities.

Specifically,
let $v$ be a vertex of $G$. By the Chernoff bound (Appendix, Theorem~\ref{thm:Chernoff}),
	it follows that,
	with high probability it holds that
\begin{equation}
	n_{1} := |N(v)|=d\pm\tilde{O}(\sqrt{d}).\label{eq:N_1}
\end{equation}
Throughout the proof, we condition on the
	event that Equation~\ref{eq:N_1} holds.
	
	Given $n_{1}$ and a vertex $u\notin N(v)$,
	the probability	that $u$ is in $\tilde{N}_{2}(v)$ is \linebreak $p_{2}=1-(1-\frac{d}{n-1})^{n_{1}}$.
	Let $n_{2} := |\tilde{N}_{2}(v)|$.
	Then, the following hold:
	\begin{enumerate}
		\item If $d=o(n^{\frac{1}{2}})$ then, by Bernoulli's inequality (Appendix, Theorem~\ref{thm:Bernoulli}),
		the quantity $p_{2}$ is $\frac{n_{1}d}{n}+(\frac{n_{1}d}{n})^2\approx \frac{n_{1}d}{n}$.
        Notice that we use $\approx$ to denote the omission of low order terms;
        further, when $a \approx b$ for some values $a$ and $b$,
        we say that $a$ is \emph{approximately} $b$.
        It follows that $\mathbb{E}_{G\sim G_{n,d}}[n_{2}]\approx(n-d)\frac{d^{2}}{n}\approx d^{2}$.  
		\item If $d=\omega(n^{\frac{1}{2}})$ then the quantity $p_{2}$
		is approximately $1$. It follows that $\mathbb{E}_{G\sim G_{n,d}}[n_{2}]\approx n$.
		\item If $d=\Theta(n^{\frac{1}{2}})$ then the quantity $p_{2}$ is some constant $c_{d,2}$ that depends on $d$.
		Specifically, if $d=cn^{1/2}$ then $c_{d,2} \approx 1-e^{-\frac{n_{1}d}{n}}=1-e^{-c}$.
		It follows that	$\mathbb{E}_{G\sim G_{n,d}}[n_{2}]=c_{d,2}n$.
	\end{enumerate}
	
\noindent	Applying Chernoff bound, we conclude that,
with high probability it holds that
	\begin{equation}
	n_{2}=\mathbb{E}[n_{2}]\pm\tilde{O}(\sqrt{d}).\label{eq:N_2}
	\end{equation}
	Throughout the proof, we condition
	on the event that Equation~\ref{eq:N_2} holds.
	
	Let $n_{i}=|\tilde{N}_{i}(v)|$. By repeating the above arguments
	for $2<i\le\log n$, \linebreak
	we conclude that the following hold with high probability:

	\begin{enumerate}
		\item If $d=o(n^{\frac{1}{i}})$ then $n_{i}=d^{i}\pm\tilde{O}(\sqrt{d^{i}})$.
		\item If $d=\omega(n^{\frac{1}{i}})$ and $d=o(n^{\frac{1}{i-1}})$
		then $|N_{i}(v)|=n-\tilde{O}(\sqrt{n})$.
		\item If $d=\Theta(n^{\frac{1}{i}})$ then $n_{i}=c_{d,i}n\pm\tilde{O}(\sqrt{n})$
		some some constant $c_{d,i}$.
	\end{enumerate}
	
	\noindent Applying union bound, we conclude that
	$n_{i}$ is as stated above
	for every $i$ and every vertex $v$ in $G$.
	Next we consider two cases,
	differentiated by the value of~$d$.
	
	\medskip\noindent\textbf{Case 1.}
	Let $d=\omega(n^{\frac{1}{i}})$ and also $d=o(n^{\frac{1}{i-1}})$ for some $i$.
	Consider two players,	placing their facilities on two arbitrary vertices, $v_1$ and $v_2$.
	By the analysis above,
	we have that
	$
	|\tilde{N}_{i}(v_{k})|=n-\Theta(d^{i-1}).
	$
	Thus, it holds that
	\begin{align}
		|\{ w\,:\, d_{G}(w,v_1)=d_{G}(w,v_2)\} | & \ge|\tilde{N}_{i}(v_1)\bigcap\tilde{N}_{i}(v_2)|\label{eq:B2} \\
		& =n-|\bar{\tilde{N}}_{i}(v_1)\bigcup\bar{\tilde{N}}_{i}(v_2)|\nonumber\\
		& \ge n-|\bar{\tilde{N}}_{i}(v_1)|-|\bar{\tilde{N}}_{i}(v_2)|\nonumber \\
		& =n-\Theta(d^{i-1}).\nonumber
	\end{align}

\noindent
We conclude that the score each player gets is at least $n / 2 - \Theta(d^{i - 1})$ and at most $n / 2 + \Theta(d^{i - 1})$,
which,
for large enough values of $n$,
means that in this case,
all placements are $\delta n$-balanced.

   \medskip
   \begin{remark}
     Let us briefly explain how the proof generalizes to $k = \Theta(d)$.
     Here,
     we would have obtained the following:
   	\begin{align*}|\{ w\,:\, d_{G}(w,v_{1})=d_{G}(w,v_{2})=...=d_{G}(w,v_{k})\} | & \ge|\bigcap_{j=1}^{k}\tilde{N}_{i}(v_{j})|\\
   	& =n-|\bigcup_{j=1}^{k}\bar{\tilde{N}}_{i}(v_{j})|\\
   	& \ge n-\sum_{j=1}^{k}|\bar{\tilde{N}}_{i}(v_{j})|\\
   	& =n-kd^{i-1}.
   	\end{align*}
   	For the graph to be,
   	say,
   	$0.1n$-balanced, the value of $k$ has to satisfy $k d^{i-1}\le 0.1n$.
   	Since $d^{i-1}=\frac {n}{d}$, it follows that $k\le 0.1d$. 	
  \end{remark}

  \noindent\textbf{Case 2.}
	Let $d=\Theta(n^{\frac{1}{i}})$ for some $i$.
	Consider two players,	placing their facilities on two arbitrary vertices, $v_1$ and $v_2$.
	Then, with high probability the following holds (using similar techniques, one can generalize the proof for this case to general $k$ as well):
		\begin{align}
		|\{ w\,:\, d_{G}(w,v_1)<d_{G}(w,v_2)\} | & =\sum_{j=1}^{i+1}|\{ w\,:\, d_{G}(w,v_1)<d_{G}(w,v_2)\,\wedge\, w\in\tilde{N}_{j}(v_1)\} |\nonumber\\
		& \le \left(\sum_{j=1}^{i-1}|\tilde{N}_{j}(v_1)|\right) + |\tilde{N}_{i}(v_1)\setminus\tilde{N}_{i}(v_2)|
		+|\tilde{N}_{i+1}(v_1)\setminus\tilde{N}_{i+1}(v_2)|  \nonumber \\
		& \le \Theta(d^{i-1}) + |\tilde{N}_{i}(v_1)\setminus\tilde{N}_{i}(v_2)|
		+|\tilde{N}_{i+1}(v_1)\setminus\tilde{N}_{i+1}(v_2)|.\nonumber
		\end{align}
		Recall that $|\tilde{N}_i(v_1)| = |\tilde{N}_i(v_2)| \pm \tilde{O}(\sqrt{d})$,
		and therefore it holds that:
		$$|\tilde{N}_{i}(v_1)\setminus\tilde{N}_{i}(v_2)|=|\tilde{N}_{i}(v_2)\setminus\tilde{N}_{i}(v_1)|
		\pm\tilde{O}(\sqrt{d}),$$
		 and that
		 $$|\tilde{N}_{i+1}(v_1)\setminus\tilde{N}_{i+1}(v_2)|=|\tilde{N}_{i+1}(v_2)\setminus\tilde{N}_{i+1}(v_1)|
		 \pm\tilde{O}(\sqrt{d}).$$

		  It follows that the score difference between the players is $\Theta(d^{i-1})+\tilde{O}(\sqrt{d})$,
		  which,
		  for sufficiently large values of $n$,
		  is smaller than $\delta n$,
		  as needed.
	%
~\qed\end{proof}

\begin{remark}
Notice that some restriction on the value of $k$ in the statement of the last theorem is needed;
for example,
if $k > d$, then $d$-regular graphs generally do not achieve reasonable balancedness.
To see this,
notice that for $k = d + 1$,
placing one facility on an arbitrary vertex $v$ while placing other facilities on all its $d$ neighbors as well
results in a score of $1$ for the facility placed on $v$
(as long as the graph is connected).
\end{remark}

\section{Expanders}\label{section:expanders}

The analysis performed in the last section shows that random graphs have good balancedness properties.
Moreover,
taking a closer look at our proof,
it looks as if the main property of random graphs which we used is that random graphs are good expanders.
 Namely, it is
 a well known fact that with high probability $\lambda_{2}\left(G\right)=\Theta\left(\sqrt{d}\right)$
 (for example, see \cite{feige2005spectral,furedi1981eigenvalues,friedman1989second}).

This motivates further studying whether expanders have good balancedness properties,
which is the subject of the current section.
It turns out that,
in some sense,
the balancedness of expanders is inferior to that of random graphs,
as is apparent from the degrees restriction imposed in the statement of Theorem~\ref{thm:expanders are balanced},
which is backed by an example of an expander with bad balancedness properties,
depicted in Figure~\ref{figure:expander which is not balanced}
and described in the proof of Theorem~\ref{thm:expander which is not balanced}.

The structure of the next proof,
which is deferred to the appendix,
is somewhat similar to the structure of the proof of Theorem~\ref{thm:random graphs are balanced},
and the overall tactics is to argue that the sizes of the $i$th neighborhoods of most graphs are roughly similar.
The main difference between the proof of the next theorem and the proof of Theorem~\ref{thm:random graphs are balanced} is that,
while the proof of Theorem~\ref{thm:random graphs are balanced} uses
probabilistic arguments to estimate the neighborhood sizes of various facilities,
the proof of the next theorem uses expansion properties of expanders.

We mention that in Section~\ref{section:algorithmtwo}
we build upon the next theorem for devising a somewhat,
a-priori surprising algorithm for verifying graph balancedness,
based on spectral analysis.

\begin{restatable}{theorem}{theoremexpanders}\label{thm:expanders are balanced}
  Let $G$ be a $d$-regular $\lambda$-expander graph of $n$ vertices,
  let $\lambda = O(\sqrt{d})$, and let $d = n^\epsilon$.
  Then, $G$ is $o(n)$-balanced for $k = o(d)$ players,
  as long as there is no positive integer $i$ for which $\epsilon = \frac{1}{i}$.
\end{restatable}

\begin{remark}
We mention that Theorem~\ref{thm:expanders are balanced} can be extended to graphs which are roughly-regular expanders.
While in the proof of Theorem~\ref{thm:expanders are balanced} we used the expander mixing lemma for regular graphs,
there exists a general version for the expander mixing lemma which gives,
for irregular graphs,
similar statements as those we used (e.g., see,~\cite{chung1997spectral}).
Using the general form of the expander mixing lemma for roughly-regular graphs,
the extension for roughly-regular expanders follows.
\end{remark}

Next we demonstrate that the restriction on $\epsilon$ in Theorem~\ref{thm:expanders are balanced} is necessary
(showing that it is not an artifact of our proof technique).
Specifically,
we show a family of expander graph whose degree regularity does not follow the aforementioned restriction,
and such that these graphs have bad balancedness properties.

\begin{theorem}\label{thm:expander which is not balanced}
  Let $d=\Theta (n^{\frac{1}{i}})$.
  Then,
  for every positive integer $i$,
  there exist roughly $d$-regular $\lambda$-expander graphs with $n$ vertices that are not $o(n)$-balanced.
\end{theorem}

\begin{proof}
	We show an example for the case where $d=\sqrt{n}$, and mention that this example can be generalized for $d=\Theta (n^{\frac{1}{i}})$.
	Let $G$ be a random graph drawn by the distribution $G_{n,\sqrt{n}}$ (see Section~\ref{section:preliminaries}).
	We construct the graph $G'$ as follows:
	
	\begin{enumerate}[label=\textbf{Step \arabic*.}, leftmargin=*]
	\item Add a new \emph{root} vertex $r$ as well as another $\sqrt{n}$ \emph{new} vertices to the vertex set of $G$. 	
	\item By introducing new edges, connect each new vertex to $\sqrt{n}$ original vertices such that no two new vertices share a common neighbor.
	\item By introducing new edges, connect $r$ to each of the $\sqrt{n}$ new vertices.	
	\end{enumerate}
		
	First we prove that $G'$ is not $o(n)$-balanced.
	It follows from the proof of Theorem~\ref{thm:random graphs are balanced} that,
	with high probability, it holds for any \emph{original vertex} (that is any vertex of $G$), that
	$|N_2(v)|\approx (1-e^{-1})n$
	(the addition of $\sqrt{n}+1$ new vertices effects $|N_2(v)|$ by at most $\sqrt{n}$).
	However, $|N_2 (r)|=n$ by construction.
	Thus, $G'$ is not $o(n)$-balanced;
	indeed, placing one facility on $r$ while the other on original vertices results in a non $o(n)$-balanced placement.
	
	Next we provide an upper bound on $\lambda (G')$, to show that $G'$ is indeed an expander, as claimed.
	Naturally,
	we rely on the fact that random graphs are good expanders.
	That is,
	it is well known that,
	with high probability,
	it holds that $\lambda_{2}(G)=\Theta(\sqrt{d})$
	(see, for example,~\cite{feige2005spectral,friedman1989second,furedi1981eigenvalues}).
	Next we show that our modifications to $G$ do not change its expansion too much.
	
	First,
	notice that in Step~1 we added $\sqrt{n}+1$ eigenvectors to the graph, all with eigenvalue zero.
	To see this,
	for each added vertex $u$, consider the vector $\vec{1}_u$ that has value $1$ on $u$'s coordinate and $0$ on the rest of the coordinates.
	Thus,
	after performing the modification described in Step~1 the expansion of the graph do not change.
	
	Second,
	consider the following inequality from perturbation theory for matrices
	that holds for any two symmetric matrices $A,N\in\mathbb{R}^{n,n}$
	(see, for example, \cite{bhatia2013matrix}):
	\begin{equation}
	\max_{i:1\le i\le n}|\lambda_{i}(A+N)-\lambda_{i}(A)|\le\max_{i:1\le i\le n}|\lambda_{i}(N)|.\label{eq:Sums of Eigenvalues}
	\end{equation}
	Namely, this inequality shows that adding a matrix $N$ to a matrix	$A$ can change the eigenvalues of $A+N$
	by at most $\max_{i:1\le i\le n}|\lambda_{i}(N)|$.
	Notice that we can write the adjacency matrix of $G'$ as
	$
	A_{G'}=A_{G}+A_{S}+A_{S_{\sqrt{n}}},
	$
	where $A_S$ is added in Step~2 and $A_{S_{\sqrt{n}}}$ is added in Step~3.
	Here, the graph $S_{m}$ is the star graph with $m+1$ vertices
	(recall that a star graph is a bipartite graph with $m+1$ vertices where one vertex is connected to all the other vertices)
	and $S$ is the disjoint union of $\sqrt{n}$ copies of $S_{\sqrt{n}}$.
	
	Finally,
	we rely on the fact that, for star graphs $S_{m}$ with $m + 1$ vertices,
	it holds that $|\lambda_{i}(S_{m})|\le\sqrt{m}$, for every $i \in [m + 1]$.
	Hence, since $S$ is the disjoint union of star graphs, for every $i$ it follows that
	$|\lambda_{i}(A_{S})|\le n^{\frac{1}{4}}$.
	Applying Inequality~\ref{eq:Sums of Eigenvalues},
	for every $i$, we have that
	\begin{align*}
		|\lambda_{i}(A_{G})|-2n^{\frac{1}{4}} & \le|\lambda_{i}(A_{G'})|\label{eq:G_H is expander} \le|\lambda_{i}(A_{G})|+2n^{\frac{1}{4}}.\nonumber
	\end{align*}

  Thus,
  we conclude that $|\lambda_{2}(G')| \le |\lambda_{2}(G)|+\Theta(\sqrt{|\lambda_{2}(G)|})$.
	Further, the above construction gives a roughly $d$-regular graph; thus, the proof follows.
~\qed\end{proof}

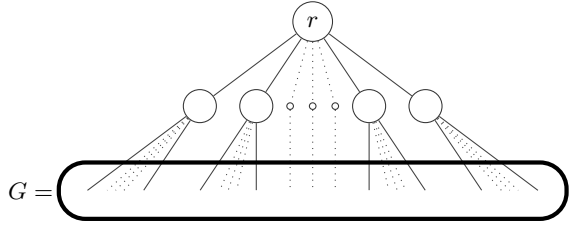
\begin{figure}[t]
    \center
    \begin{tikzpicture}[draw=black!75, scale=0.75,-]
      \tikzstyle{rootvertex}=[circle,draw=black!80,minimum size=15pt,inner sep=0pt]
      \tikzstyle{newvertex}=[circle,draw=black!80,minimum size=12.5pt,inner sep=0pt]
      \tikzstyle{kill}=[circle,draw=black!80,minimum size=10pt,inner sep=0pt]
      \tikzstyle{dummyvertex}=[circle,draw=black!80,minimum size=2.5pt,inner sep=0pt]
      \tikzstyle{kill}=[circle,draw=black!80,minimum size=2.5pt,inner sep=0pt]
      \tikzstyle{kill}=[circle,draw=black!80,minimum size=0pt,inner sep=0pt]

      \foreach [count=\i] \pos / \text / \style in {
        {(4,3)}/$r$/rootvertex,
        {(2,1.5)}/$ $/newvertex,
        {(3,1.5)}/$ $/newvertex,
        {(5,1.5)}/$ $/newvertex,
        {(6,1.5)}/$ $/newvertex,
        {(0,0)}/$ $/kill,
        {(1,0)}/$ $/kill,
        {(2,0)}/$ $/kill,
        {(3,0)}/$ $/kill,
        {(5,0)}/$ $/kill,
        {(6,0)}/$ $/kill,
        {(7,0)}/$ $/kill,
        {(8,0)}/$ $/kill,
        {(3.6,1.5)}/$ $/dummyvertex,
        {(4,1.5)}/$ $/dummyvertex,
        {(4.4,1.5)}/$ $/dummyvertex,
        {(3.6,0)}/$ $/kill,
        {(4,0)}/$ $/kill,
        {(4.4,0)}/$ $/kill,
        {(0.35,0)}/$ $/kill,
        {(0.5,0)}/$ $/kill,
        {(0.65,0)}/$ $/kill,
        {(2.35,0)}/$ $/kill,
        {(2.5,0)}/$ $/kill,
        {(2.65,0)}/$ $/kill,
        {(5.35,0)}/$ $/kill,
        {(5.5,0)}/$ $/kill,
        {(5.65,0)}/$ $/kill,
        {(7.35,0)}/$ $/kill,
        {(7.5,0)}/$ $/kill,
        {(7.65,0)}/$ $/kill}
      {
        \node[\style] (V\i) at \pos {\text};
      }

      \foreach \i / \j in {1/2,1/3,1/4,1/5,2/6,2/7,3/8,3/9,4/10,4/11,5/12,5/13} {
        \path[] (V\i) edge (V\j);
      }

      \foreach \i / \j in {1/14,1/15,1/16,14/17,15/18,16/19,2/20,2/21,2/22,3/23,3/24,3/25,4/26,4/27,4/28,5/29,5/30,5/31} {
        \draw[dotted] (V\i) edge (V\j);
      }

      \node[] () at (-1, 0) {$G = $};
      
      \draw [ultra thick, rounded corners=10pt, draw=black, fill opacity=0.2]
       (-0.5,-0.5) -- (8.5,-0.5) -- (8.5,0.5) -- (-0.5, 0.5) -- cycle;

    \end{tikzpicture}
    \caption{An expander which is not $\delta n$-balanced,
    as described and proved in Theorem~\ref{thm:expander which is not balanced}.
    The bottom oval represents the graph $G$,
    the middle vertices are the $\sqrt{n}$ new vertices,
    and the top vertex $r$ is the root vertex.}\label{figure:expander which is not balanced}
  \end{figure}

\section{Verifying Balancedness}\label{section:algorithms}

In this section we first give some evidence that deciding a given graph's balancedness is intractable
(see Theorem~\ref{theorem:hardness}).
We view this result as motivating the design of algorithms which can verify, with high probability,
whether a given graph is balanced.
Concretely,
we use the analysis described in Section~\ref{section:randomgraphs} and Section~\ref{section:expanders}
to devise two algorithms which produce \emph{certificates} of balancedness:
  the first algorithm,
  described in Section~\ref{section:algorithmone},
  is based on graph traversal,
  while the second algorithm,
  described in Section~\ref{section:algorithmtwo},
  is based on spectral analysis,
  and results in better running time compared to the first algorithm.

\para{Intractability Result.}\label{section:intractability}
Consider the following decision problem.

\vspace{-5px}
\probDef
  {Unbalancedness}
  {A graph $G$, number of facilities $k$, and a bound $s$.}
  {Is there a placement of $k$ facilities on $G$ such that the score of at least one of the players is strictly less than $s$?}
\vspace{-5px}

\textsc{Unbalancedness} can be solved in $O(n^{k + 3})$ time,
as follows.
First, compute all-pairs shortest paths on $G$
(this can be done in $O(n^3)$ time, using Floyd-Warshall algorithm, for example).
Then,
iterate over all ${n \choose k}$ possible placements of $k$ facilities on $G$,
computing the balancedness of each such placement using the all-pairs shortest paths matrix.

It is natural to ask whether there are better algorithms than the simple algorithm described above.
The next theorem, whose proof is deferred to the appendix, shows that this is not the case (assuming P $\neq$ NP).

\begin{restatable}{theorem}{theoremhardness}\label{theorem:hardness}
  \textsc{Unbalancedness} is NP-hard.
\end{restatable}

\begin{remark}
Taking a closer look at the proof of Theorem~\ref{theorem:hardness},
one might notice that in the reduction from \textsc{Dominating Set} to \textsc{Unbalancedness},
the number of players $k$ in the \textsc{Unbalancedness} instance is $O(h)$,
where $h$ is the size of the dominating set in the \textsc{Dominating Set} instance.
Since it is known that,
assuming the Exponential Time Hypothesis (ETH),
no algorithm running in $n^{o(k)}$ time exists for \textsc{Dominating Set}~\cite{eth1},
it follows that no algorithm running in $n^{o(k)}$ time exists for \textsc{Unbalancedness}.
\end{remark}

\subsection{Graph Traversal Algorithm}\label{section:algorithmone}

In this section we describe an algorithm which verifies whether a given graph is balanced.
Specifically,
if the algorithm accepts then it generates a \emph{certificate},
such that the graph is guaranteed to be balanced.
Otherwise,
if the algorithm rejects,
then there are no balancedness guarantees;
notice that it does not mean that the graph is not balanced,
as it might be the case that the graph is balanced but the algorithm is not able to produce a certificate.
The algorithm runs in $O(n (m + n))$,
where $n$ is the number of vertices in the graph and $m$ is the number of edges in the graph;
thus,
from the intractability result proved in Theorem~\ref{theorem:hardness},
we have no reason to believe that the algorithm would produce a certificate for all graphs.
Using the analysis described in the proof of Theorem~\ref{thm:random graphs are balanced},
however,
we will show that,
with high probability,
the algorithm is able to produce certificates for random graphs.

\begin{theorem}\label{theorem:algorithmone}
  Let $d$ be such that there is no positive integer $i$ for which $d = n^{\frac{1}{i}}$
  and let $k = c_1 d$ for some constant $c_1$. 
  Then,
  there is an algorithm running in $O(n (m + n))$ time which,
  given a random graph $G \sim G_{n,d}$ with $m$ edges,
  with high probability produces a \emph{certificate} of its $\delta n$-balancedness for $k$ players
  where $\delta \geq c_1 c_2$ and $c_2$ is some constant.
\end{theorem}

\begin{proof}
Let us denote the radius of G by $r$;
that is,
$r = \min_{v \in V} \max_{u \in V} d(u, v)$.
The algorithm first computes the value of $|N_i(v)|$ for each $v \in V$ and for each $i \in [r]$;
to compute these values,
it is sufficient to traverse the graph once from each vertex,
say, by performing a breadth-first search.
Then,
the algorithm accepts iff the following conditions hold:

\begin{enumerate}[label=\textbf{Condition \arabic*.}, leftmargin=*]

\item
$|N_r(v)| = n$ for each $v \in V$.

\item
$|N_{r - 1}(v)| \leq \delta n / k$ for each $v \in V$.

\end{enumerate}

It might be useful to think about this algorithm as generating a table $T$ (or a matrix) of $n$ rows and $r$ columns,
where $n$ is the number of vertices in the graph and $r$ is the radius of the graph.
The value in the cell $T[v, i]$ is the size of the $i$th neighborhood of vertex $v$;
that is, $T[v, i] := |N_i(v)|$.
Given the table $T$,
the algorithm checks,
in Condition 1,
that the values in the $r$th column (that is, in the ``last'' column) are all exactly $n$;
in other words, for the given graph it holds that $r = \max_{v \in V} \max_{u \in V} d(u, v)$,
so each vertex reaches the whole graph on $r$ hops.
Then,
in Condition 2,
the algorithm checks that the values in the $(r - 1)$th column are at most $\delta n / k$;
in other words, for the given graph it holds that each vertex reaches at most $\delta n / k$ vertices in $r - 1$ hops.

By analysis similar to that in the proof of Theorem~\ref{thm:random graphs are balanced},
we have that if these conditions hold,
indeed the graph is $\delta n$ balanced.
To see why,
notice that combining both conditions,
for any $k$ players,
we have that all players reach at least $n - \delta n$ vertices in the $r$th hop,
and on these vertices, all players get the same fraction of $1 / k$.
Thus,
the score of each player is at least $(n - \delta n) / k \geq n / k - \delta n$ and at most $\delta n + (n - \delta n) / k \leq \delta n + n / k$.
It follows that the graph is indeed $\delta n$-balanced.

Taking a closer look at Case 1 in the proof of Theorem~\ref{thm:random graphs are balanced},
it follows first that $|N_r(v)| = n$ holds for every vertex $v$ with high probability,
thus Condition 1 holds.
Further,
for every vertex satisfies $|N_{r - 1}(v)| \leq c_2 n / d$ for some constant $c_2$ with high probability.
Therefore,
\[
  |N_{r - 1}(v)| \leq \frac{c_2 n}{d} = \frac{c_1 c_2 n}{k} \leq \frac{\delta n}{k}
\]
holds with high probability for every vertex $v$,
thus Condition 2 holds.
~\qed\end{proof}

\subsection{Spectral Algorithm}\label{section:algorithmtwo}

In this section we describe another algorithm which verifies whether a given graph is balanced.
In contrast to the previous section, here we provide a  \emph{probabilistic certificate}.
Specifically,
we consider a randomized algorithm $A$ with the following promise:
  if,
  for a given graph $G$ as an input,
  algorithm $A$ accepts with probability larger than $0.9$,
  then $G$ is balanced.
We say that such an algorithm is a \emph{probabilistic certificate}.
Probabilistic certificates are useful since,
given a graph~$G$,
one could estimate the accepted probability by repeatedly running $A$.

Notice that there might be ``false negatives'',
as there could be balanced graphs with low acceptance probability.
What we do show next is that for a random graph $G$ drawn from $G_{n,d}$,
with high probability over the distribution $G_{n,d}$,
the graph $G$ is such that $A$ accepts $G$
(with probability larger than $0.9$ over the random bits $A$ uses).

Algorithm $A$ runs in $\tilde{O}(nd)$ where $d$ is the expected degree of $G$.
Thus,
following the intractability result proved in Theorem~\ref{theorem:hardness},
we have no reason to believe that the algorithm can produce certificates for all graphs.
Using the analysis of Theorem~\ref{thm:expanders are balanced},
however,
we will show that,
with high probability,
the algorithm is able to produce certificates for (most of the) random graphs.

\begin{theorem}\label{theorem:algorithmone}
  Let $d$ be such that there is no positive integer $i$ for which \linebreak $d = n^{\frac{1}{i}}$.
  There is an algorithm running in $\tilde{O}(nd)$ time which,
  given a random graph $G \sim G_{n,d}$,
  with high probability produces a \emph{probabilistic certificate} of its $\delta$-balancedness for some $\delta = o(n)$.
\end{theorem}

\begin{proof}
Let $\epsilon = 0.01$.
We say that $a$ is an $\epsilon$-estimation of $b$ if it holds that $a \leq (1 + \epsilon)b$.
Let $A_G$ denote the adjacency matrix of a graph $G$.
Given a graph $G$ as an input, the algorithm operates as follows.

\begin{enumerate}[label=\textbf{Step \arabic*.}, leftmargin=*]

\item
Set $d$ to be equal to the expected degree in $G$. 

\item
If the graph is not roughly-regular then reject.

\item
Calculate an $\epsilon$-estimation of the second largest-in-magnitude eigenvalue of $A_G$ and denote its value by $\tilde{\lambda_2}$.

\item
If $\tilde{\lambda_2}$ is smaller than $100\sqrt{d}$ then accept, otherwise reject.

\end{enumerate}
	
We start with a running time analysis of the above algorithm.
Clearly, Steps 1,2 and 4 can be accomplished in $O(nd)$ time.
For Step 3, we shall specify how to calculate $\tilde{\lambda_2}$.
Fortunately,
the problem of calculating $\tilde{\lambda_2}$ of a given graph is a well-studied problem which can be solved using the power method:
  the power method is a randomized algorithm that, using $t = \Omega (log n)$ operations of multiplying certain vectors by $A_G$,
  gives with high probability (say $p=0.9$) the desired $\epsilon$-estimation of the second largest-in-magnitude eigenvalue of $A_G$.
Given any vector $v$,
calculating $A_G \cdot v$ can be done in $O(dn)$ time
(to see this note that $G$ has roughly $\frac{d}{2}n$ edges, hence $A_G$ has roughly $dn$ non zero entries).
In total, Step 3 takes $tnd=\tilde{O}(dn)$ time.
For more details on the power method,
see, e.g.,~\cite[Lemma 8.1]{vishnoi2012laplacian}.
	
By standard calculations,
one can show that random graphs are roughly-regular graphs (see definition in Section~\ref{section:expanders}).
Hence,
with high probability $G$ is such that it (always) passes Step~1.
It is a well known fact that with high probability $\lambda_{2}\left(G\right)=\Theta\left(\sqrt{d}\right)$
(see~\cite{feige2005spectral,furedi1981eigenvalues,friedman1989second}, for example).
Hence, once again, with high probability $G$ is such that it will be accepted in Step~4 with probability $p$.
By Theorem~\ref{thm:expanders are balanced},
we have that if $G$ is a roughly-regular spectral expander,
then $G$ is $\delta n$ balanced; thus, we are done.
~\qed\end{proof}

\section{Outlook}\label{section:outlook}

We briefly discuss several directions for future research.

\para{Robustness of balanced graphs.}
Random graphs might be seen as noisy complete graphs,
in the sense that,
starting from a complete graph,
each edge is removed with a certain probability ($1 - p$).
Since complete graphs are $0$-balanced for any number of players $k$,
one might understand Theorem~\ref{thm:random graphs are balanced} as proving
that complete graphs are not only $0$-balanced for any $k$,
but are also \emph{robust} to noise.
It is not clear whether other families of graphs with good balancedness properties
are also robust to noise.

\para{Cooperative, randomized, and adversarial tie-breaking.}
In this paper we broke ties evenly,
which has similar behavior as to breaking ties uniformly at random;
that is,
if $z$ facilities are the closest to a vertex,
then each of the corresponding players gets a score increase of $1 / z$.
It is natural,
and practically motivated,
to study other tie-breaking schemes,
such as
(1) breaking ties in favor of the balancedness;
(2) breaking ties at random by some distribution over the players;
(3) breaking ties against balancedness.

\para{Best-case, average-case, and worst-case balancedness.}
In this paper we concentrated on a worst-case notion of balancedness,
by defining a graph to be balanced iff \emph{each} placement on it is balanced.
What happens if we consider an average-case notion of balancedness,
by requiring a placement on it to be balanced with high probability,
or even a best-case notion of balancedness,
by requiring that at least one balanced placement shall exist?
We observe that some graphs are not balanced even for the best-case notion;
for example,
notice that \emph{any} placement of two facilities on the graph depicted in Figure~\ref{figure:unbalanced} is not $0$-balanced
(the figure and a proof of the claim are given in the Appendix).

\paragraph*{Acknowledgements.}

We thank Uriel Feige for useful discussions.

\bibliographystyle{plainurl}
\bibliography{bib}

\begin{thebibliography}{10}

\bibitem{ahn2004competitive}
Hee-Kap Ahn, Siu-Wing Cheng, Otfried Cheong, Mordecai Golin, and Rene
  Van~Oostrum.
\newblock Competitive facility location: the {V}oronoi game.
\newblock {\em Theoretical Computer Science}, 310(1-3):457--467, 2004.

\bibitem{alon1988explicit}
Noga Alon and Fan~RK Chung.
\newblock Explicit construction of linear sized tolerant networks.
\newblock {\em Annals of Discrete Mathematics}, 38:15--19, 1988.

\bibitem{bandyapadhyay2015voronoi}
Sayan Bandyapadhyay, Aritra Banik, Sandip Das, and Hirak Sarkar.
\newblock Voronoi game on graphs.
\newblock {\em Theoretical Computer Science}, 562:270--282, 2015.

\bibitem{banik2013optimal}
Aritra Banik, Bhaswar~B. Bhattacharya, and Sandip Das.
\newblock Optimal strategies for the one-round discrete {V}oronoi game on a
  line.
\newblock {\em Journal of Combinatorial Optimization}, 26(4):655--669, 2013.

\bibitem{banik2016discrete}
Aritra Banik, Jean-Lou De~Carufel, Anil Maheshwari, and Michiel Smid.
\newblock Discrete voronoi games and $\epsilon$-nets, in two and three
  dimensions.
\newblock {\em Computational Geometry}, 55:41--58, 2016.

\bibitem{baroni2015fixed}
Enrico Baroni, Remco van~der Hofstad, and Julia Komjathy.
\newblock Fixed speed competition on the configuration model with infinite
  variance degrees: {U}nequal speeds.
\newblock {\em Electronic Journal of Probability}, 20, 2015.

\bibitem{bhatia2013matrix}
Rajendra Bhatia.
\newblock {\em Matrix analysis}, volume 169.
\newblock 2013.

\bibitem{chernoff1952measure}
Herman Chernoff.
\newblock A measure of asymptotic efficiency for tests of a hypothesis based on
  the sum of observations.
\newblock {\em The Annals of Mathematical Statistics}, pages 493--507, 1952.

\bibitem{chung1997spectral}
Fan~RK Chung.
\newblock {\em Spectral graph theory}, volume~92.
\newblock American Mathematical Soc., 1997.

\bibitem{durr2007nash}
Christoph D{\"u}rr and Nguyen~Kim Thang.
\newblock Nash equilibria in {V}oronoi games on graphs.
\newblock In {\em Proceedings of ESA 2007}, pages 17--28. 2007.

\bibitem{feige2005spectral}
Uriel Feige and Eran Ofek.
\newblock Spectral techniques applied to sparse random graphs.
\newblock {\em Random Structures \& Algorithms}, 27(2):251--275, 2005.

\bibitem{friedman1989second}
Joel Friedman, Jeff Kahn, and Endre Szemeredi.
\newblock On the second eigenvalue of random regular graphs.
\newblock In {\em Proceedings of the twenty-first annual ACM symposium on
  Theory of computing}, pages 587--598. ACM, 1989.

\bibitem{Friesz2007}
Terry~L. Friesz.
\newblock Competitive facility location.
\newblock {\em Networks and Spatial Economics}, 7(1):1--2, 2007.

\bibitem{furedi1981eigenvalues}
Zolt{\'a}n F{\"u}redi and J{\'a}nos Koml{\'o}s.
\newblock The eigenvalues of random symmetric matrices.
\newblock {\em Combinatorica}, 1(3):233--241, 1981.

\bibitem{GJ79}
Michael~R. Garey and David~S. Johnson.
\newblock {\em Computers and Intractability: A Guide to the Theory of
  {NP}-Completeness}.
\newblock 1979.

\bibitem{godsil2013algebraic}
Chris Godsil and Gordon~F. Royle.
\newblock {\em Algebraic graph theory}, volume 207.
\newblock 2013.

\bibitem{honiden2009balancing}
Shinichi Honiden, Michael~E. Houle, and Christian Sommer.
\newblock Balancing graph voronoi diagrams.
\newblock In {\em Proceedings of the Sixth International Symposium on Voronoi
  Diagrams (ISVD 2009)}, pages 183--191, 2009.

\bibitem{marin2011discrete}
Alfredo Mar{\'\i}n.
\newblock The discrete facility location problem with balanced allocation of
  customers.
\newblock {\em European Journal of Operational Research}, 210(1):27--38, 2011.

\bibitem{marsh1994equity}
Michael~T. Marsh and David~A. Schilling.
\newblock Equity measurement in facility location analysis: A review and
  framework.
\newblock {\em European Journal of Operational Research}, 74(1):1--17, 1994.

\bibitem{mavronicolas2008voronoi}
Marios Mavronicolas, Burkhard Monien, Vicky~G. Papadopoulou, and Florian
  Schoppmann.
\newblock Voronoi games on cycle graphs.
\newblock In {\em Mathematical Foundations of Computer Science 2008}, pages
  503--514. Springer, 2008.

\bibitem{eth1}
Mihai P{\u{a}}tra{\c{s}}cu and Ryan Williams.
\newblock On the possibility of faster {S}{A}{T} algorithms.
\newblock In {\em Proceedings of SODA 2010}, pages 1065--1075, 2010.

\bibitem{saban2012competitive}
Daniela Saban and Nicolas Stier-Moses.
\newblock The competitive facility location problem in a duopoly: {C}onnections
  to the 1-median problem.
\newblock In {\em Proceedings of WINE 2012}, pages 539--545, 2012.

\bibitem{teramoto2006voronoi}
Sachio Teramoto, Erik~D. Demaine, and Ryuhei Uehara.
\newblock Voronoi game on graphs and its complexity.
\newblock In {\em Computational Intelligence and Games, 2006 IEEE Symposium
  on}, pages 265--271. IEEE, 2006.

\bibitem{vadhan2012pseudorandomness}
Salil~P. Vadhan.
\newblock {\em Pseudorandomness}.
\newblock Now, 2012.

\bibitem{van2015fixed}
Remco van~der Hofstad and Julia Komjathy.
\newblock Fixed speed competition on the configuration model with infinite
  variance degrees: {E}qual speeds.
\newblock {\em arXiv preprint arXiv:1503.09046}, 2015.

\bibitem{vishnoi2012laplacian}
Nisheeth~K. Vishnoi.
\newblock Laplacian solvers and their algorithmic applications.
\newblock {\em Theoretical Computer Science}, 8(1-2):1--141, 2012.

\end{thebibliography}

\appendix

\newpage

\section{Useful Facts}

\begin{theorem}[Chernoff~\cite{chernoff1952measure}]\label{thm:Chernoff}
  Let $\{X_{1},\ldots,X_{n}\}$ be independent random variables with $0 \leq X_{i} \leq 1$, for each $i$.
  Let $X = \sum_{i=1}^{n} X_{i}$,
  $\mu = \mathbb{E} [ X ]$,
  and $\sigma^{2} = VAR ( X )$.
  Then,
  it holds that
  \[
	\Pr[|X-\mu|\ge\delta\sigma]\le C\max\{ \exp(-c\delta^{2}),\exp(-c\delta\sigma)\},
  \]
  for some absolute constants $c$ and $C$,
  for every $\delta \geq 0$.
\end{theorem}

\begin{theorem}[Bernoulli's Inequality]\label{thm:Bernoulli}
  The following claims hold.

  \begin{enumerate}

  \item
  For $x \geq -1$ and $r > 1$,
  it holds that
  \[
    (1 + x)^r \geq 1 + xr.
  \]

  \item
  For $x \in [ -1, 1 / (r - 1) )$ and $r > 1$,
  it holds that
  \[
    (1 + x)^r \leq 1 + \frac{rx}{1 - (r - 1) x}.
  \]

  \end{enumerate}

\end{theorem}

\begin{lemma}[{Spectral to Vertex Expansion~\cite[Chapter 4]{vadhan2012pseudorandomness}}]\label{lem:Spectral to vertex expansion}
  Let $G = (V, E)$ be a $d$-regular $\lambda$-expander graph with $n$ vertices.
  Then,
  for every $\alpha \in [ 0, 1 ]$,
  it holds that $G$ is an $( \alpha n, \frac{1}{( 1 - \alpha ) ( \frac{\lambda}{d} )^2 + \alpha} )$
  vertex expander.
\end{lemma}

\begin{lemma}[Expander Mixing Lemma~\cite{alon1988explicit}]\label{lem:Mixing Lemma}
  Let $G = (V, E)$ be a $d$-regular $\lambda$-expander graph with $n$ vertices.
  Then,
  it holds that
  \[
	|E_{G}(S,T)-\frac{d\cdot|S|\cdot|T|}{n}|\leq\lambda\sqrt{|S|\cdot|T|},
  \]
  for any two,
  not necessarily disjoints subsets $S ,T \subseteq V$.
\end{lemma}

\section{Missing Proofs}

We provide proofs missing from the main text.

\subsection{Proof of Theorem~\ref{thm:expanders are balanced}}

\theoremexpanders*

\begin{proof}
The proof structure is somewhat similar to the proof of Theorem~\ref{thm:random graphs are balanced},
and the overall tactics is to argue that the sizes of the $i$th neighborhoods of most graphs are roughly similar.
The main difference between the current proof and the proof of Theorem~\ref{thm:random graphs are balanced} is that,
while the proof of Theorem~\ref{thm:random graphs are balanced} uses probabilistic arguments to estimate the neighborhood sizes of various facilities,
here we naturally use the expanders' expansion properties.
For simplicity we present the proof for $k = 2$.

Let $v$ be a vertex of $G$. Since $G$ is $d$-regular, it follows by definition that
$n_{1} := |N(v)| = d$.
Further,
let $n_{2} := |\tilde{N}_{2}(v)|$,
and consider the following cases.

\begin{enumerate}

	\item If $\epsilon < \frac{1}{2}$,
	then $n_{2}=\Omega ( n_{1}d)$,
	by Lemma~\ref{lem:Spectral to vertex expansion} (which we apply by specifically taking $\alpha = 1/\sqrt{n}$).
	Since the graph is $d$-regular, it follows that $n_{2}=\Theta ( n_{1}d)$.
		
	\item If $\epsilon > \frac{1}{2}$, then $n_2=n-o(n)$.
	This holds since, by applying the expander mixing lemma (Lemma~\ref{lem:Mixing Lemma}),
	we have that
	\begin{align*}
		 \frac{d}{n}|\bar{N}_2(v)|d\le\lambda\sqrt{|\bar{N}_2(v)|d}.
	\end{align*}
	Thus, $$|\bar{N}_2(v)|=\Theta ((\frac{n}{d})^2) = o(n),$$ and, indeed, $n_2=n-o(n)$.
	
	\item Notice that the case where $d=\Theta(n^{\frac{1}{2}})$ is excluded explicitly in the theorem statement.
\end{enumerate}

\noindent Let $n_{i}=|\tilde{N}_{i}(v)|$. By repeating the above arguments
for $2<i\le\log n$, we have the following:
\begin{enumerate}
	\item If $d=o(n^{\frac{1}{i}})$ then $n_{i}=\Theta (d^{i})$.
	\item If $d=\omega(n^{\frac{1}{i}})$ and $d=o(n^{\frac{1}{i-1}})$
	then $|N_{i}(v)|=n-o(n)$.
	\item Again, the case where $d=\Theta(n^{\frac{1}{i}})$ is excluded.
\end{enumerate}

\noindent
Notice that $n_{i}$ is as above,
for every $i$ and every vertex $v$ in $G$.

\medskip

By the assumption on $\epsilon$ it holds that  $\epsilon > \frac{1}{i}$ and $\epsilon < \frac{1}{i-1}$
for some $i$. Consider two players, placing their facilities on arbitrary vertices $v_{0}$ and $v_{1}$.
It follows that
\[
\tilde{N}_{i}(v_0) = \tilde{N}_{i}(v_1)=n-\Theta(d^{i-1})
\]
and therefore
\begin{align}
|\{ w\,:\, d_{G}(w,v_0)=d_{G}(w,v_1)\} | & \ge|\tilde{N}_{i}(v_0)\bigcap\tilde{N}_{i}(v_1)|\label{eq:B21} \\
& =n-|\bar{\tilde{N}}_{i}(v_0)\bigcup\bar{\tilde{N}}_{i}(v_0)|\nonumber\\
& \ge n-|\bar{\tilde{N}}_{i}(v_0)|-|\bar{\tilde{N}}_{i}(v_1)|\nonumber \\
& =n-\Theta(d^{i-1}).\nonumber
\end{align}

Using the above equation,
we obtain that the score of each player is at least $n/2 - \Theta(d^{i-1})$ and at most $n/2 + \Theta(d^{i-1})$,
which,
for sufficiently large values of $n$,
means that the graph is $\delta n$-balanced, as needed.
~\qed\end{proof}

\subsection{Proof of Theorem~\ref{theorem:hardness}}

\theoremhardness*

\begin{figure}[t]
    \center
    \begin{tikzpicture}[draw=black!75, scale=1.8,-]
      \tikzstyle{original}=[circle,draw=black!80,minimum size=20pt,inner sep=0pt]
      \tikzstyle{bagvertex}=[circle,draw=black!80,minimum size=40pt,inner sep=0pt]
      \tikzstyle{vvertex}=[circle,draw=black!80,minimum size=50pt,inner sep=0pt]

      \foreach [count=\i] \pos / \text / \style in {
        {(2.5,3)}/$v$/vvertex,
        {(1.5,2)}/$v_1$/original,
        {(2,1)}/$v_2$/original,
        {(3,1)}/$v_3$/original,
        {(3.5,2)}/$v_4$/original,
        {(0.5,2)}/$\{v_1^b : b \in [n^3]\}$/bagvertex,
        {(1,0.5)}/$\{v_2^b : b \in [n^3]\}$/bagvertex,
        {(4,0.5)}/$\{v_3^b : b \in [n^3]\}$/bagvertex,
        {(4.5,2)}/$\{v_4^b : b \in [n^3]\}$/bagvertex}
      {
        \node[\style] (V\i) at \pos {\text};
      }

      \foreach \i / \j in {1/2,1/3,1/4,1/5,2/3,3/4} {
        \path[] (V\i) edge (V\j);
      }

      \foreach \i / \j in {2/6,2/7,3/6,3/7,3/8,4/7,4/8,5/9} {
        \draw[line width=2] (V\i) edge (V\j);
      }

    \end{tikzpicture}
    \caption{An example of the reduction described in the proof of Theorem~\ref{theorem:hardness}.
    The input graph for the \textsc{Dominating Set} instance is composed of the vertices $v_1$, $v_2$, $v_3$, and $v_4$,
    which are connected by the edges $\{v_1, v_2\}$ and $\{v_2, v_3\}$.
    $h$ is set to $2$.
    The figure shows the reduced instance to \textsc{Unbalancedness}
    where a \textbf{bold} edge connecting an original vertex to a bag represents $n^3$ edges,
    connecting the original vertex to all the corresponding bag vertices.}\label{figure:reduction}
  \end{figure}
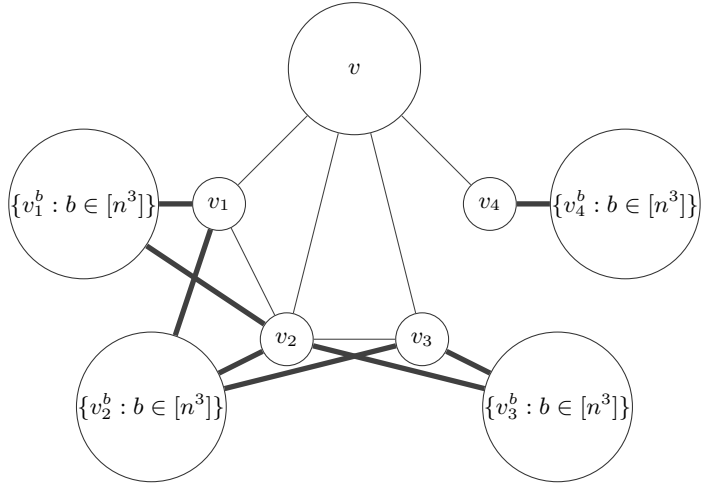

\begin{proof}
A vertex $v$ is \emph{dominated} by a set of vertices $S$ if it is adjacent to at least one vertex from $S$.
A \emph{dominating set} is a set of vertices which dominates all other vertices in the graph.
We describe a reduction from the following NP-hard problem~\cite{GJ79}.

\probDef
  {Dominating Set}
  {A graph $G$ and a budget $h$.}
  {Is there a dominating set of $G$ consisting of $h$ vertices?}

Given an instance of \textsc{Dominating Set},
we create an instance of \textsc{Unbalancedness},
as follows.
Let us denote the vertices of $G$ by $v_1, \ldots, v_n$
and refer to them as \emph{original vertices}.
For each original vertex $v_i$ ($i \in [n]$)
we create $n^3$ new vertices,
denoted by $v_i^b$ ($i \in [n]$, $b \in [n^3]$)
and refer to them as \emph{bag vertices};
we make each such bag a clique, that is,
for each $i \in [n]$, we add all edges $\{v_i^{b'}, v_i^{b''}\}$, for $b' \neq b'', b' \in [n^3], b'' \in [n^3]$.
We connect each original vertex $v_i$ to its own bag vertices~$v_i^b$
(that is, we add all edges $\{v_i, v_i^b\}$, for $b \in [n^3]$)
as well as to the bag vertices of its adjacent vertices
(that is, if the edge $\{v_i, v_j\}$ is present, then we add all edges $\{v_i, v_j^b\}$, for $b \in [n^3]$).
We create a new vertex, denoted by $v$,
and connect it to all the original vertices (that is, we add the edges $\{v, v_i\}$, for $i \in [n]$).
We set the number $k$ of facilities to $h + 1$ and the bound~$s$ to $n$.
This finishes the polynomial-time reduction.
An example of the construction is given in Figure~\ref{figure:reduction}.

Next we prove correctness.
Given a dominating set $u_1, \ldots, u_h$,
we place one facility on $v$ and the other $k - 1 = h$ facilities on the dominating set vertices $u_1, \ldots, u_h$.
The score of the facility placed on $v$ is at most $n + 1 - h \leq n$,
since it gets no score from the bag vertices and the dominating set vertices.
Thus, there exist a player whose score is strictly less than $s$.

For the other direction,
Consider a placement of $k$ facilities on the constructed graph such that there exists at least one player, say $P_1$, whose score is less than $s = n$.
We claim that $P_1$'s facility must be placed on $v$:
  if this is not the case,
  that is,
  if $P_1$ is placing a facility on either an original vertex or a bag vertex,
  then it has distance one to at least one bag of $n^3$ bag vertices;
  as a result, $P_1$'s score would be at least $n^2 > n$.

Following the last paragraph,
we have that $P_1$ is placing a facility on $v$.
Notice that, roughly speaking, our goal now is to reduce the score of $P_1$ to the minimum possible.
First, we claim that no player can place a facility on a bag vertex.
In contradiction,
let us assume some player $P_i$ places a facility on a bag vertex $v_i^b$,
for some $i \in [n]$ and $b \in [n^3]$.
Then,
if another facility is placed on its corresponding original vertex $v_i$,
then we pick an arbitrary original vertex with no facility placed on it,
and move $P_i$'s facility to this arbitrarily-chosen original vertex.
The score of $P_1$ will not increase as a result of this move.
Otherwise,
if no facility is placed on the original vertex $v_i$ corresponding to the vertex $v_i^b$,
then we move $P_i$'s facility to its corresponding original vertex $v_i$.
Again,
the score of $P_1$ will not increase as a result of this move.
Thus,
we can assume without loss of generality that no facility is placed on a bag vertex.

Thus,
we have that $P_1$ is placing a facility on $v$,
and all other players are placing facilities on original vertices.
Now,
if those facilities do not form a dominating set,
then at least one original vertex $v_i$ exists which is not adjacent to any facility.
The facility placed by $P_i$ on $v$ will have distance of $2$ to the bag vertices $v_i^b$ ($b \in [n^3]$) corresponding to the undominated vertex $v_i$,
and all other facilities will have distance of at least $2$ to these vertices.
Thus,
the score of $v$ will be at least $n^3 / k \geq n^2 > s$,
therefore,
we conclude that the other facilities are placed on vertices which form a dominating set.
~\qed\end{proof}

\subsection{Proof of Proposition~\ref{proposition unbalanced graph}}

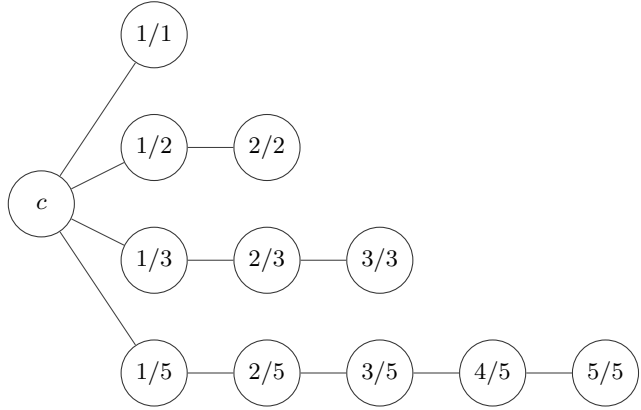
\begin{figure}[t]
    \center
    \begin{tikzpicture}[draw=black!75, scale=1.5,-]
      \tikzstyle{vertex}=[circle,draw=black!80,minimum size=25pt,inner sep=0pt]
	
	  \footnotesize

      \foreach [count=\i] \pos / \text in {
        {(0,1.5)}/$c$,
        {(1,3)}/${1/1}$,
        {(1,2)}/${1/2}$,
        {(2,2)}/${2/2}$,
        {(1,1)}/${1/3}$,
        {(2,1)}/${2/3}$,
        {(3,1)}/${3/3}$,
        {(1,0)}/${1/5}$,
        {(2,0)}/${2/5}$,
        {(3,0)}/${3/5}$,
        {(4,0)}/${4/5}$,
        {(5,0)}/${5/5}$}
      {
        \node[vertex] (V\i) at \pos {\text};
      }

      \foreach \i / \j in {1/2,1/3,3/4,1/5,5/6,6/7,1/8,8/9,9/10,10/11,11/12} {
        \path[] (V\i) edge (V\j);
      }
    \end{tikzpicture}
    \caption{A graph which is not $0$-balanced for any placement of two facilities.}\label{figure:unbalanced}
  \end{figure}

\begin{proposition}\label{proposition unbalanced graph}
  For any placement of two facilities,
  the graph depicted in Figure~\ref{figure:unbalanced} is not $0$-balanced.
\end{proposition}

\begin{proof}
Consider the graph depicted in Figure~\ref{figure:unbalanced},
composed of a center vertex $c$ and four branches:
  first branch, consisting of vertex $1/1$;
  second branch, consisting of vertices $1/2$, $2/2$;
  third branch, consisting of vertices $1/3$, $2/3$, and $3/3$,
  and fourth branch, consisting of vertices $1/5$, $2/5$, $3/5$, $4/5$, and $5/5$.
We argue that the graph depicted in Figure~\ref{figure:unbalanced} is not balanced for any placement of two facilities on it.
To this end,
let us denote one facility by $f_1$ and the other by $f_2$,
and consider the following case analysis.

\smallskip\noindent\textbf{$f_1$ is placed on $c$:}
in this case,
the best position for $f_2$ is ${1/5}$,
making $f_1$'s score strictly greater than that of $f_2$.

\smallskip\noindent\textbf{$f_1$ and $f_2$ are placed on the same branch:}
in this case,
the facility which is placed closer to $c$ has strictly larger than the other facility.

\smallskip\noindent\textbf{$f_1$ and $f_2$ are placed on different branches:}
we split this case into two subcases.
First,
if both facilities are placed such that they have the same distance to $c$,
then the facility which is placed on the longer branch wins.
Second,
if the facilities are placed such that one of them has smaller distance to $c$,
then,
without loss of generality,
assume that the distance between $f_1$ and $c$ is strictly smaller than the distance between $f_2$ and $c$.
Then,
the best placement for $f_2$ is when it is the closest to all the vertices of the longest path,
thus getting a score of $5$.
However,
even in this case $f_1$ has strictly larger score, specifically $7$.
~\qed\end{proof}

\end{document}